\documentclass[11pt]{article}
\usepackage{fullpage}
\usepackage{amsfonts, amsthm, amssymb, amsmath}
\usepackage[pdftex]{graphicx,color}
\usepackage[scaled=0.92]{helvet}
\usepackage{times}

\usepackage[pdftex, pagebackref=false, colorlinks=true, urlcolor=black, citecolor=black, linkcolor=black]{hyperref}

\usepackage{framed}

\usepackage{verbatim}
\usepackage{fullpage}

\usepackage{epsfig}

\usepackage{verbatim}

\newtheorem{thm}{Theorem}[section]
\newtheorem{lem}[thm]{Lemma}
\newtheorem{cor}[thm]{Corollary}
\newtheorem{defn}[thm]{Definition}
\newtheorem{clm}[thm]{Claim}
\newtheorem{cons}[thm]{Construction}

\newenvironment{theorem}{\begin{thm}\begin{rm}}%
{\end{rm}\end{thm}}
\newenvironment{lemma}{\begin{lem}\begin{rm}}%
{\end{rm}\end{lem}}
{\end{rm}\end{cor}}
{\end{em}\end{defn}}
{\end{rm}\end{clm}}
{\end{em}\end{cons}}

\newcommand{\secref}[1]{\hyperref[#1]{Section \ref{#1}}}
\newcommand{\apref}[1]{\hyperref[#1]{Appendix \ref{#1}}}
\newcommand{\thref}[1]{\hyperref[#1]{Theorem \ref{#1}}}
\newcommand{\defref}[1]{\hyperref[#1]{Definition \ref{#1}}}
\newcommand{\corref}[1]{\hyperref[#1]{Corollary \ref{#1}}}
\newcommand{\lemref}[1]{\hyperref[#1]{Lemma \ref{#1}}}
\newcommand{\clref}[1]{\hyperref[#1]{Claim \ref{#1}}}
\newcommand{\consref}[1]{\hyperref[#1]{Construction \ref{#1}}}
\newcommand{\figref}[1]{\hyperref[#1]{Figure \ref{#1}}}
\newcommand{\eqnref}[1]{\hyperref[#1]{Equation \ref{#1}}}

\newcommand{\BSTCONN}{{\sc{Balanced ST-Connectivity}}}

\title{{\bf{A Note on the Balanced ST-Connectivity}}}
\vspace{0.15in}
\author{
Shiva Kintali\thanks{Department of Computer Science, Princeton University, Princeton, NJ 08540. Email : {\em{kintali@cs.princeton.edu}}}\and
Asaf Shapira\thanks{School of Mathematics and College of Computing, Georgia Institute of Technology, Atlanta, GA-30332. Email : {\em{asafico@math.gatech.edu}}}
}
\date{}

\begin{document}

\maketitle

\begin{abstract}
We prove that every {\sf YES} instance of \BSTCONN\ \cite{kintali-real-nlvsl} has a balanced path of polynomial length.
\end{abstract}


\section{Introduction}

Kintali \cite{kintali-real-nlvsl} introduced new kind of connectivity problems called {\em{graph realizability problems}}, motivated by the study of AuxPDAs \cite{cook-auxpda}. In this paper, we study one such graph realizability problem called \BSTCONN\ and prove that every {\sf YES} instance of \BSTCONN\ has a balanced path of polynomial length.

Let $\mathcal{G}(V,E)$ be a directed graph and let $n = |V|$. Let $\mathcal{G'}(V,E')$ be the underlying undirected graph of $\mathcal{G}$. Let $P$ be a path in $\mathcal{G'}$. Let $e = (u,v)$ be an edge along the path $P$. Edge $e$ is called {\em{neutral}} edge if both $(u,v)$ and $(v,u)$ are in $E$. Edge $e$ is called {\em{forward}} edge if $(u,v) \in E$ and $(v,u) \notin E$. Edge $e$ is called {\em{backward}} edge if $(u,v) \notin E$ and $(v,u) \in E$.

A path (say $P$) from $s \in V$ to $t \in V$ in $\mathcal{G'}(V,E')$ is called {\em{balanced}} if the number of forward edges along $P$ is equal to the number of backward edges along $P$. A balanced path might have any number of neutral edges. By definition, if there is a balanced path from $s$ to $t$ then there is a balanced path from $t$ to $s$.

\begin{framed}
\noindent \BSTCONN\ : Given a directed graph $\mathcal{G}(V,E)$ and two distinguished nodes $s$ and $t$, decide if there is {\em{balanced path}} between $s$ and $t$.
\end{framed}

A balanced path may not be a simple path. The example in Figure \ref{fig:bstconn-n2} shows an instance of \BSTCONN\ where the {\em{only}} balanced path between $s$ and $t$ is of length $\Theta(n^2)$. The directed simple path from $s$ to $t$ is of length $n/2$. There is a cycle of length $n/2$ at the vertex $v$. All the edges (except $(v,u)$) in this cycle are undirected. The balanced path from $s$ to $t$ is obtained by traversing from $s$ to $v$, traversing the cycle clockwise for $n/2$ times and then traversing from $v$ to $t$.

\clearpage

\begin{figure}[htp]
\begin{center}
\includegraphics[width=4in]{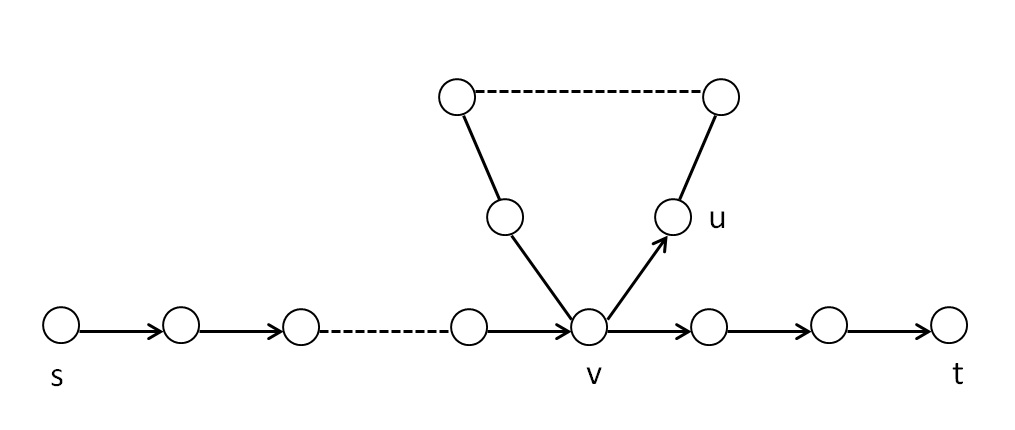}
\end{center}
\caption{A non-simple balanced path of length $\Theta(n^2)$ from $s$ to $t$}
\label{fig:bstconn-n2}
\end{figure}

\section{Length of Balanced Paths}

We now prove that every {\sf YES} instance of \BSTCONN\ has a balanced path of polynomial length. We need the following lemma.

\begin{lemma}\label{lem:coeffbound}
Let $c_1 < c_2 < \dots < c_r \in [n]$ and $k \in [n]$. If $m_1, m_2, \dots, m_r$ are integers such that
\begin{center}
$m_1c_1 + m_2c_2 + \dots + m_rc_r = k$,
\end{center}
then there exist integers $m'_1, m'_2, \dots,m'_r$ satisfying
\begin{center}
$m'_1c_1 + m'_2c_2 + \dots + m'_rc_r = k$
\end{center}
such that $|m'_1|+|m'_2|+\dots+|m'_r| \leq O(nr)$.
\end{lemma}

\begin{proof}
Let $a_{i} = \lfloor{\frac{m_{i}}{c_{r}}}\rfloor$ and $m_i = a_ic_r+b_i$ for $1 \leq i \leq r-1$. We have,
\begin{center}
$(a_1c_r+b_1)c_1 + (a_2c_r+b_2)c_2 + \dots + (a_{r-1}c_r+b_{r-1})c_{r-1}+ m_rc_r = k$
\end{center}
Rearranging we get,
\begin{center}
$b_1c_1 + b_2c_2 + \dots + b_{r-1}c_{r-1} + (m_r+a_1c_1+a_2c_2+\dots+a_{r-1}c_{r-1})c_r = k$.
\end{center}
Note that $|b_i| < c_r < n$ for $1 \leq i \leq r-1$. Hence, $b_1c_1 + b_2c_2 + \dots + b_{r-1}c_{r-1} - k = O(n\sum_{i=1}^{r-1}{c_{i}})$. Hence, $m_r+a_1c_1+a_2c_2+\dots+a_{r-1}c_{r-1} = O(n\sum_{i=1}^{r-1}{c_{i}})/c_{r} = O(nr)$.

Setting $m'_i = b_i$ for $1 \leq i \leq r-1$ and $m'_{r} = m_r+a_1c_1+a_2c_2+\dots+a_{r-1}c_{r-1}$ we get the desired result.
\end{proof}

\begin{theorem}\label{thm:balanced-n2}
Let $G(V,E)$ be a directed graph with two distinguished vertices $s,t \in V$ and let $P$ be a balanced path from $s$ to $t$. Then there exists a balanced path $Q$ from $s$ to $t$ such that the length of $Q$ is $O(n^3)$.
\end{theorem}

\begin{proof}

We decompose $P$ into a simple path (say $P'$) from $s$ to $t$ and a set of cycles $\mathcal{C} = \{C_1,C_2,\dots,C_l\}$. Let $c_1,\dots,c_r$ be the distinct lengths of the cycles in $\mathcal{C}$. Let $-k$ denote the number of forward edges minus the number of backward edges along $P'$ from $s$ to $t$. Since there is a balanced path from $s$ to $t$ using the path $P'$ and the cycles from $\mathcal{C}$, there exist integers $m_1,\dots,m_r$ satisfying $m_1c_1+\dots+m_rc_r=k$. Applying Lemma \ref{lem:coeffbound} there exist integers $m'_1,\dots,m'_r$ satisfying $m'_1c_1+\dots+m'_rc_r=k$ such that $|m'_1|+\dots+|m'_r| \leq O(nr)$.

We now construct a balanced path $Q$ from $s$ to $t$ as follows : For every $m'_i$ we walk $m'_i$ times around the cycle of length $c_i$ (if there are several cycles of this length, we choose one of them arbitrarily). Note that these cycles may not be connected to each other. We now choose an arbitrary vertex from each cycle and connect it to $t$ by simple paths (say $P_1, P_2, \dots, P_r$).

The new balanced path $Q$ starts from $s$ and follows the simple path $P'$ from $s$ to $t$ and uses $P_i$ to reach the cycle of length $c_i$ and walks around it $m'_i$ times and comes back to $t$ using $P_{i}$. This is repeated for $1 \leq i \leq r$. Since each $P_i$ is used once while going away from $t$ and once while coming back to $t$, the paths $P_1, P_2, \dots, P_r$ do not modify the balancedness of the path $Q$. The combined length of paths $P_1, P_2, \dots, P_r$ is $O(nr)$. Since $|m'_1|+\dots+|m'_r| \leq O(nr)$ the overall length of the balanced path $Q$ is $O(nr) = O(n^{2})$.
\end{proof}

\bibliographystyle{alpha}
\bibliography{bib-balanced}

\end{document}